\documentclass[a4paper,american,10pt]{article}
\usepackage{graphicx}
\usepackage{amsmath}
\usepackage{amssymb}
\usepackage{amsthm}
\usepackage{float}

\makeatletter

\newcommand{\myaddress}[1]{
\par {\raggedright #1
\vspace{1.4em}
\noindent\par}
}

\theoremstyle{plain}
\newtheorem{thm}{Theorem}

\newtheorem{lem}[thm]{Lemma}
\newtheorem{prop}[thm]{Proposition}

\makeatother

\begin{document}

\title{Efficiency of quantum state tomography for qubits}

\author{Koichi Yamagata}

\maketitle

\myaddress{\begin{center}
Department of Mathematics, Osaka University \\
Toyonaka, Osaka 560-0043, Japan
\par\end{center}}

\def\H{\mathcal{H}}
\def\K{\mathcal{K}}
\def\M{\mathcal{M}}
\def\B{\mathcal{B}}
\def\L{\mathcal{L}}
\def\E{\mathcal{E}}
\def\X{\mathfrak{X}}
\def\P{\mathcal{P}}
\def\F{\mathcal{F}}
\def\T{\mathcal{T}}
\def\N{\mathbb{N}}
\def\R{\mathbb{R}}
\def\C{\mathbb{C}}
\def\braket#1#2{\langle#1|#2\rangle}
\def\inner#1#2{\left\langle #1,\,#2\right\rangle }
\def\bra#1{\langle#1|}
\def\ket#1{|#1\rangle}
\def\S{\mathcal{S}}
\def\Tr{{\rm Tr}\,}
\def\rank{{\rm rank}\,}
\def\ipovm{\widetilde{\M}(\H)}
\def\realmatrix{M^{+}(d,\mathbb{R})}
\def\supp#1{|{\rm supp}(#1)|}
\def\argmax{\mathop{\rm argmax}}
\def\argmin{\mathop{\rm argmin}}
\def\Ker{{\rm Ker}\,}
\def\Re{{\rm Re}\,}
\def\cov{{\rm cov}}
\def\spann{\rm{span}}
\def\v{\boldsymbol{v}}

\begin{abstract}
The efficiency of quantum state tomography is discussed from the point of view of 
quantum parameter estimation theory, in which the trace of the weighted covariance is to be minimized. 
It is shown that tomography is optimal only when a special weight is adopted. 
\end{abstract}

\section{Introduction}

Let $\L(\H)$ be the set of linear operators on a Hilbert space $\H=\C^{2}$,
and let 
$\S := \{\tau_x \mid x = (x^{1},x^{2},x^{3}) \in \X \}$ 
be the set of strictly positive density operators on $\H$ parametrized by the Stokes parameters 
$x \in \X := \{x \in \R^{3} \mid (x^{1})^{2}+(x^{2})^{2}+(x^{3})^{2}<1\}$ as
\begin{equation}
\tau_x:=\frac{1}{2} (I+x^{1}\sigma_{1}+x^{2}\sigma_{2}+x^{3}\sigma_{3}) \label{eq:tomoModel},
\end{equation}
where 
$\sigma_1, \sigma_2, \sigma_3$ are the Pauli matrices. 
Suppose we have an unknown quantum state $\tau=\tau_x \in \S$.
We are interested in identifying the true value of the parameter $x$. 

Let
\[
\M^{(s)}(\H):=\{(M_{1},M_{2},\dots,M_{s}) \mid M_{i}\in\L(\H),\, M_{i}\geq0,\,\sum_{i=1}^{s}M_i=I\}
\]
be the set of positive operator-valued measures (POVMs) on $\H$ taking values on a finite set of outcomes labeled by $\{1,2, \dots ,s  \}$,
and let $\M(\H)=\bigcup_{s=1}^{\infty}\M^{(s)}(\H)$. 
Given POVMs $M=(M_{1},M_{2}, \dots , M_{s_1})$, $N=(N_{1},N_{2},\dots,N_{s_2})$, and a real number $p$ between $0$ and $1$, 
we can generate a new POVM by a randomized combination of them as follows:
$$ pM \oplus (1-p)N := (pM_{1}, \dots , pM_{s_1},(1-p)N_{1}, \dots , (1-p)N_{s_2}) \in \M(\H). $$
We can repeat this randomization procedure inductively to obtain  
$ \bigoplus_{i=1}^{k} p_{i}M^{(i)} \in \M(\H)$, where 
$M^{(1)},M^{(2)}, \dots , M^{(k)} \in \M(\H)$ and $p_{i} \geq 0 \, (1 \leq i \leq k)$
such that $ \sum_{i=1}^{k} p_{i} = 1$. 
We shall call $\bigoplus_{i=1}^{k}p_{i}M^{(i)}$ a {\em random measurement}
when $M^{(1)},M^{(2)},\dots,M^{(k)}\in\M(\H)$ are all projection-valued measurements (PVMs). 
Applying a random measurement means applying one of the projection-valued measurement $\{M^{(i)}\}_{1 \leq i \leq k}$ chosen at random according to the probability distribution $p=(p_i)_{1 \leq i \leq k}$\footnote{
Helstrom \cite{cramer} defined a random measurement based on a different type of convex structure of $\M(\H)$ as $(p M_{1}+(1-p) N_{1},\dots,p M_{s}+(1-p) N_{s})$.
Our definition of random measurement is seemingly different from his.}.

Let $M^{(1)},M^{(2)},M^{(3)}$
be projection-valued measurements given by the spectral decomposition of the observables $\sigma_{1},\sigma_{2},\sigma_{3}$, respectively, 
and let $M^{(T)}:=\frac{1}{3}(M^{(1)}\oplus M^{(2)}\oplus M^{(3)})$ be their random measurement according to the uniform distribution. 
Suppose that, among $m$ applications of $M^{(T)}$ to the unknown state $\tau_x$, the $\mu$th PVM $M^{(\mu)}$ has been chosen $m_{\mu}$ times 
and the outcomes $ \pm 1$ have been observed $m_{\mu}^{\pm}$ times, 
where $m=m_1 + m_2 + m_3$ and $m_{\mu}= m_{\mu}^{+} + m_{\mu}^{-}$ for $\mu \in \{1,2,3\}$. 
We can construct an unbiased estimator for the Stokes parameters $x=(x^1,x^2,x^3)$ as
\begin{equation}
\hat{x}^{\mu}  :=  \frac{m_{\mu}^{+}-m_{\mu}^{-}}{m_{\mu}}, \qquad \mu \in \{1,2,3\} \label{eq:tomo}.
\end{equation}
We shall call this estimator a {\em tomography} in this paper.
Note that the tomography can be regarded as a maximum likelihood estimator. 
In fact, since the probability distribution for the outcomes $\pm 1$ of the $\mu$th PVM 
\begin{equation}
M^{(\mu)}  =  \left(\frac{1}{2}(I+\sigma_{\mu}),\,\frac{1}{2}(I-\sigma_{\mu})\right),  \label{eq:pvmTomo}
\end{equation}
applied to the state $\tau_x \in \S$ is given by
$p_{\tau_x}^{M^{(\mu)}}=(\frac{1+x^{\mu}}{2},\frac{1-x^{\mu}}{2})$,
the probability distribution for the outcome of the tomography $M^{(T)}$ is
\begin{equation}
p_{\tau_x}^{M^{(T)}}=\frac{1}{6}(1+x^{1},1-x^{1},1+x^{2},1-x^{2},1+x^{3},1-x^{3}). \label{eq:pM}
\end{equation}
As a consequence, 
the likelihood function for the outcomes $(m_{\mu}^{\pm})_{1 \leq \mu \leq 3}$ obtained by $m$ applications of $M^{(T)}$ is 
$$l_{m}(x) = \sum_{\mu=1}^{3}\left(m_{\mu}^{+}\log\frac{1+x^{\mu}}{6}+m_{\mu}^{-}\log\frac{1-x^{\mu}}{6}\right),$$
and it is easy to see that $\frac{\partial}{\partial x^{\mu}}l_{m}=0$ is equivalent\footnote{
There are possibilities that $\hat{x} \notin \X$. 
However it follows from the law of large numbers of the tomography 
that $\hat{x} \in \X$ for sufficiently large $m$ almost surely. } to \eqref{eq:tomo}.

In order to investigate the optimality of the tomography,
let us recall some basic facts from quantum parameter estimation theory.
Let $\{\rho_{\theta} \mid \theta=(\theta^1,\dots,\theta^d)\in\Theta\}$ be a smooth parametric family of density operators
on a Hilbert space $\H$ with parameter space $\Theta \subset \R^d$. 
An estimator is represented by a pair $(M,\hat{\theta})$ of a POVM $M \in \M(\H)$ and a map $\hat{\theta}:\N \rightarrow \Theta$ 
that gives the estimated value $\hat{\theta}(n)$ from each observed data $n\in\N$.
An estimator $(M,\hat{\theta})$ is called {\em unbiased} if
\begin{equation}
E_{\theta}[M,\hat{\theta}]:=\sum_{n\in\N} \hat{\theta}(n)\, \Tr \rho_{\theta} M_n=\theta \label{eq:unbias}
\end{equation}
is satisfied for all $\theta \in \Theta$. 
An estimator $(M,\hat{\theta})$ is called {\em locally unbiased} \cite{unbias} at a given point $\theta_0 \in \Theta$ 
if the condition \eqref{eq:unbias} is satisfied around $\theta=\theta_0$ up to the first order of the Taylor expansion.
It is well known that an estimator $(M,\hat{\theta})$ that is locally unbiased at $\theta_0$ satisfies 
the following series of inequalities
\cite{unbias,cramer}:
\begin{equation}
V_{\theta_0}[M,\hat{\theta}] \geq (g_{\theta_0}(M))^{-1} \geq (J_{\theta_0})^{-1}, \label{eq:cramer}
\end{equation}
where $V_{\theta}[\cdot ]$ denotes the covariance matrix, 
and $g_{\theta}(M)$ is the classical Fisher information matrix at $\theta$ with respect to $M\in\M(\H)$ defined by
\[
	g_{\theta}(M):=
	\left[
	\sum_{n}\frac{(\frac{\partial}{\partial \theta^i}\Tr\rho_{\theta}M_{n})(\frac{\partial}{\partial \theta^j}\Tr\rho_{\theta}M_{n})}{\Tr\rho_{\theta}M_{n}}
	\right]
	_{1 \leq i,j \leq d}.
\]
Further, $J_{\theta}$ is the quantum Fisher information matrix at $\theta$ 
given by
$$
J_{\theta}:= \left[ \Tr (\frac{\partial}{\partial \theta^i} \rho_{\theta}) L_j  \right]_{1 \leq i,j \leq d}
=\left[ \frac{1}{2} \Tr \rho_{\theta}( L_i L_j + L_j L_i)  \right]_{1 \leq i,j \leq d},
$$
where $L_i$ is the $i$th symmetric logarithmic derivative (SLD) defined by the selfadjoint operator
satisfying the equation
\begin{equation}
\frac{\partial}{\partial \theta^i}\rho_{\theta}=\frac{1}{2}(L_{i}\rho_{\theta}+\rho_{\theta}L_{i}).\label{eq:SLD}
\end{equation}
The inequality $V_{\theta_0}[M,\hat{\theta}] \geq (J_{\theta_0})^{-1}$ is called the {\em quantum Cram\'er-Rao inequality}. 
The first inequality in \eqref{eq:cramer} is saturated 
when $\hat{\theta}^{i}(n)=\theta^i + \sum_j (g_{\theta}(M)^{-1})^{ij} \frac{\partial}{\partial \theta^j} (\log \Tr \rho_{\theta}M_n)$ 
is adopted. 
However the second inequality in \eqref{eq:cramer} cannot be saturated in general 
because of the non-commutativity of the SLDs.
To avoid this difficulty, 
we often adopt an alternative strategy to seek the estimator which minimizes $\Tr H_{\theta_0} V_{\theta_0}[M,\hat{\theta}]$, 
where $H_{\theta}$ is a given $d \times d$ real positive definite matrix for each $\theta$ called a {\em weight} \cite{unbias,cramer}. 
Thus the problem of finding the optimal estimator boils down to the problem of finding $M\in\M(\H)$ 
which minimizes $\Tr H_{\theta_0} g_{\theta_0}(M)^{-1}$.

It is known that 
when $\dim \H=2$, there is a definitive answer to the optimality of estimators, which is summarized in the following Propositions.

\begin{prop}\label{prop:qbitest}
For a given weight $H_{\theta}$, 
\begin{equation}
\min\left\{ \Tr H_{\theta}g_{\theta}(M)^{-1} \mid M\in\M(\H)\right\} =\left(\Tr R_{\theta}\right)^{2} \label{eq:qbitBound},
\end{equation}
where $R_{\theta}:=\sqrt{\sqrt{J_{\theta}^{-1}}H_{\theta}\,\sqrt{J_{\theta}^{-1}}}$.
The minimum is attained   
if and only if $M\in\M(\H)$ satisfies
\begin{equation}
g_{\theta}(M)=\frac{\sqrt{J_{\theta}}R_{\theta}\sqrt{J_{\theta}}}{\Tr R_{\theta}}.\label{eq:bestFisher}
\end{equation} 
\end{prop}

Proposition \ref{prop:qbitest} was first proved by Nagaoka \cite{nagaoka2para} (cf. \cite{fujiwara2para}) when $d=2$.
The case $d=3$ is proved by Hayashi \cite{hayashi3para}, and independently by Gill and Massar \cite{gill}.
Further, Nagaoka constructed explicitly a measurement which attains the minimum when $d=2$.
His construction of an optimal estimator can be generalized as follows.

\begin{prop}\label{prop:qbitEstM}
Given a weight $H_{\theta}$, 
let us diagonalize $R_{\theta}$ as $R_{\theta}=U S U^{-1}$
where $S=diag(S_{1},\dots,S_{d})$ is a diagonal matrix and $U\in O(d)$,
and let $M^{(i)}$ be a projection-valued measurement given by the spectral decomposition of the operator
\begin{equation}
\hat{L}^{i}:=\sum_{k=1}^{d}K^{ik}L_{k}, \label{eq:Lhat}
\end{equation}
where $K^{ik} := (U^{-1}\sqrt{J_{\theta}^{-1}})^{ik}$.
Then the random measurement
\begin{equation}
M := p_{1}M^{(1)}\oplus\dots\oplus p_{d}M^{(d)} \label{eq:optMeasurement}
\end{equation}
satisfies \eqref{eq:bestFisher}, 
where  $p_{i} := S_{i}/(S_{1}+\dots+S_{d})$. 
\end{prop}

Note that the optimal measurement \eqref{eq:optMeasurement} depends on the true value of $\theta\in\Theta$ in general.
In such a case,
we necessary invoke an adaptive estimation scheme \cite{strong} to achieve the minimum \eqref{eq:qbitBound}.

Now it is natural to inquire whether the tomography is optimal in view of Propositions \ref{prop:qbitest} and \ref{prop:qbitEstM}.
The answer is given by the following. 
\begin{thm}\label{thm:optTomo}
Tomography is optimal if and only if the weight $H_x$ is proportional to the following special one:
\begin{eqnarray}
	H_{x}^{(T)}:=
	\left(
	\begin{array}{ccc}
	\frac{1}{1-(x^1)^{2}} & -\frac{(x^1) (x^2)}{(1-(x^1)^{2})(1-(x^2)^{2})} & -\frac{(x^3)(x^1)}{(1-(x^3)^{2})(1-(x^1)^{2})}\\
	-\frac{(x^1)(x^2)}{(1-(x^1)^{2})(1-(x^2)^{2})} & \frac{1}{1-(x^2)^{2}} & -\frac{(x^2)(x^3)}{(1-(x^2)^{2})(1-(x^3)^{2})}\\
	-\frac{(x^3)(x^1)}{(1-(x^3)^{2})(1-(x^1)^{2})} & -\frac{(x^2) (x^3)}{(1-(x^2)^{2})(1-(x^3)^{2})} & \frac{1}{1-(x^3)^{2}}
	\end{array}
	\right). \label{eq:ht}
\end{eqnarray}
\end{thm}
Note that $H_{x}^{(T)}$ is not rotationally symmetric. 
This implies that the optimal weight depends on the choice of the coordinate axes. 
Theorem \ref{thm:optTomo} also implies that 
the tomography is not optimal
for a rotationally symmetric weight that is natural for a physical point of view. 

The paper is organized as follows. 
Theorem \ref{thm:optTomo} is proved in Section \ref{sec:Tomography}, and
the non-optimality of the tomography for a rotationally symmetric weight is discussed 
and numerically demonstrated in Section \ref{sec:Discussions}.  
An extension to the case when $\dim \H \geq 3$ is also discussed there. 
For the reader's convenience, simple proofs of Propositions \ref{prop:qbitest} and \ref{prop:qbitEstM} are given 
in Appendix.


\section{Proof of Theorem \ref{thm:optTomo} \label{sec:Tomography}}
We prove Theorem \ref{thm:optTomo} in a series of Lemmas.

\begin{lem}
Let $L_{\mu}$ be the SLD of $\frac{\partial}{\partial x^{\mu}}$ for $\mu \in \{1,2,3\}$. Then
$$
L_{\mu}  = \sigma_{\mu}-\frac{x^{\mu}}{2\det\tau}(I-\tau).
$$
\end{lem}

\begin{proof}
We need only verify that $L_{\mu}$ satisfies equation \eqref{eq:SLD}.
\[
L_{\mu}\tau=\sigma_{\mu}\tau-\frac{x^{\mu}}{2\det\tau}\tau(I-\tau)=\sigma_{\mu}\tau-\frac{x^{\mu}}{2}I.
\]
Therefore
\begin{eqnarray*}
\frac{1}{2}(L_{\mu}\tau+\tau L_{\mu}) & = & \frac{1}{2}(\{\tau,\sigma_{\mu}\}-x^{\mu}I)
=\frac{1}{2}(\{\frac{1}{2}I,\sigma_{\mu}\}+\{\frac{x^{\mu}}{2}\sigma_{\mu},\sigma_{\mu}\}-x^{\mu}I)\\
& = & \frac{1}{2}(\sigma_{\mu}+x^{\mu}I-x^{\mu}I)=\frac{\sigma_{\mu}}{2}=\frac{\partial}{\partial x^{\mu}}\tau
\end{eqnarray*}
where $\{A,B\}:=AB+BA$ for $A,B\in\L(\H)$.
\end{proof}
\begin{lem}\label{lem:SLDFisher}
Let $J_{x}$ be the SLD Fisher information matrix at $x$. Then
\[
J_{x}=\left(I-\ket{x} \bra{x}\right)^{-1}
\]
where $\ket{x}=\left(\begin{array}{c}
x^{1}\\
x^{2}\\
x^{3}\end{array}\right)$.
\end{lem}
\begin{proof}
We calculate the elements of $J_{x}$.
\begin{eqnarray*}
\left(J_{x}\right)_{\mu\nu} & = & \Tr\frac{\partial\tau}{\partial x^{\mu}}L_{\nu}=\Tr\frac{\sigma_{\mu}}{2}\left(\sigma_{\mu}-\frac{x^{\mu}}{2\det\tau}(I-\tau)\right)=\delta_{\mu\nu}+\frac{x^{\mu}x^{\nu}}{4\det\tau}.
\end{eqnarray*}
Thus
\[
J_{x}=I+\frac{1}{4\det\tau}\ket{x} \bra{x}=I+\frac{1}{1-r^{2}}\ket{x} \bra{x}. 
\]
Then 
\[
\left(I-\ket{x} \bra{x} \right)\left(I+\frac{1}{1-r^{2}}\ket{x} \bra{x} \right)=I+\frac{1}{1-r^{2}}\ket{x} \bra{x}-\ket{x}\bra{x}-\frac{r^{2}}{1-r^{2}}\ket{x}\bra{x}=I,\]
where $r=\sqrt{\braket{x}{x}}$. 
Therefore $I+\frac{1}{1-r^{2}}\ket{x}\bra{x}=\left(I-\ket{x}\bra{x}\right)^{-1}$.
\end{proof}

\begin{lem}\label{lem:gyaku}
Given $F_{x}\in g_{x}(\M(\H))$ with $F_{x}>0$. 
There exists a weight $H_{x}$
such that\begin{equation}
\min_{M\in\M(\H)}\left\{ \Tr H_{x}g_{x}(M)^{-1}\right\} =\Tr H_{x}F_{x}^{-1}\label{eq:cond1}\end{equation}
if and only if
\begin{equation}
\Tr J_{x}^{-1}F_{x}=1.\label{eq:cond2}
\end{equation}
Further, when \eqref{eq:cond2} is satisfied,
\begin{equation}
H_{x}=kF_{x}\, J_{x}^{-1}F_{x} \label{eq:onlyWeight}
\end{equation}
is the only weight which satisfies \eqref{eq:cond1} where k is an arbitrary real positive number. 
\end{lem}
\begin{proof}
We first assume that there exists a weight $H_{x}$ which
satisfies \eqref{eq:cond1}. 
Let $R_{x}:=\sqrt{\sqrt{J_{x}^{-1}}H_{x}\,\sqrt{J_{x}^{-1}}}$.
According to Proposition \ref{prop:qbitest}, $F_{x}$ must be 
\[
F_{x}=\frac{\sqrt{J_{x}}R_{x}\sqrt{J_{x}}}{\Tr R_{x}},
\]
so that
$$
\Tr J_{x}^{-1}F_{x} 
= \Tr J_{x}^{-1}\frac{\sqrt{J_{x}}R_{x}\sqrt{J_{x}}}{\Tr R_{x}}
=1.
$$
Then we conclude \eqref{eq:cond2}.

We next assume that \eqref{eq:cond2} is satisfied. 
Let $H_{x}=kF_{x}\, J_{x}^{-1}F$.
It follows from Proposition \ref{prop:qbitest} that
\begin{eqnarray*}
\min_{M\in\M(\H)}\Tr H_{x}g_{x}(M)^{-1} & = & \left(\Tr\sqrt{k\sqrt{J_{x}^{-1}}F_{x}\, J_{x}^{-1}F_{x}\,\sqrt{J_{x}^{-1}}}\right)^{2}\\
 & = & k\left(\Tr J_{x}^{-1}F_{x}\right)^{2}=k\left(\Tr J_{x}^{-1}F_{x}\right)\\
 & = & \Tr (kF_{x}\, J_{x}^{-1}F_{x})\, F_{x}^{-1} = \Tr H_{x}F_{x}^{-1}.
\end{eqnarray*}
Further, the weight of the form \eqref{eq:onlyWeight} are the only weights which satisfy \eqref{eq:cond1} because 
the mapping
\begin{eqnarray*}
M^{(1)}(d,\R) \ni H_{x} \mapsto \frac{\sqrt{J_{x}}\sqrt{\sqrt{J_{x}^{-1}}H_{x}\,\sqrt{J_{x}^{-1}}}\sqrt{J_{x}}}{\Tr\sqrt{\sqrt{J_{x}^{-1}}H_{x}\,\sqrt{J_{x}^{-1}}}}=\frac{\sqrt{J_{x}}R_{x}\sqrt{J_{x}}}{\Tr R_{x}}
\in g_{x}(\M(\H))
\end{eqnarray*}
is injective where $M^{(1)}(d,\R):=\{G \mid \mbox{$G$ is $d \times d$ real positive definite matrix}, \,\Tr G=1\}$.
\end{proof}

\begin{proof}[\bf{Proof of Theorem \ref{thm:optTomo}}]
We can calculate the classical Fisher information matrix with respect to $M^{(T)}$ from \eqref{eq:pM} as follow:
\begin{equation}
g_{x}(M^{(T)})=\frac{1}{3}\left(\begin{array}{ccc}
\frac{1}{1-(x^{1})^{2}} & 0 & 0\\
0 & \frac{1}{1-(x^{2})^{2}} & 0\\
0 & 0 & \frac{1}{1-(x^{3})^{2}}\end{array}\right).\label{eq:gM}
\end{equation}
Then
\begin{eqnarray*}
\Tr J_{x}^{-1}\, g_{x}(M^{(T)}) & = & \Tr\frac{1}{3}(I-\ket r\bra r)\left(
\begin{array}{ccc}
\frac{1}{1-(x^{1})^{2}} & 0 & 0\\
0 & \frac{1}{1-(x^{2})^{2}} & 0\\
0 & 0 & \frac{1}{1-(x^{3})^{2}}
\end{array}
\right)\\
 & = & \frac{1}{3}(\frac{1}{1-(x^{1})^{2}}+\frac{1}{1-(x^{2})^{2}}+\frac{1}{1-(x^{3})^{2}}-\frac{(x^{1})^{2}}{1-(x^{1})^{2}}-\frac{(x^{2})^{2}}{1-(x^{2})^{2}}-\frac{(x^{3})^{2}}{1-(x^{3})^{2}})\\
 & = & 1
\end{eqnarray*}
We see from Lemma \ref{lem:gyaku} that $H_{x}:=k\, g_{x}(M^{(T)})\, J_{x}^{-1}\, g_{x}(M^{(T)})$ are the only weights which satisfy
\[
\min_{N\in\M(\H)}\left\{ \Tr H_{x}\, g_{x}(N)^{-1}\right\} =\Tr H_{x}\, g_{x}(M^{(T)})^{-1}.
\]
Then
\begin{flalign*}
& k\, g_{x}(M^{(T)})\, J_{x}^{-1}g_{x}(M^{(T)}) \\
& =  k\, g_{x}(M^{(T)})\,(I-\ket{x}\bra{x})g_{x}(M^{(T)})=k\,(g_{x}(M^{(T)})^{2}-g_{x}(M^{(T)})\ket{x}\bra{x} g_{x}(M^{(T)}))\\
& =  9k
\left(
	\begin{array}{ccc}
	\frac{1}{1-(x^1)^{2}} & -\frac{(x^1) (x^2)}{(1-(x^1)^{2})(1-(x^2)^{2})} & -\frac{(x^3)(x^1)}{(1-(x^3)^{2})(1-(x^1)^{2})}\\
	-\frac{(x^1)(x^2)}{(1-(x^1)^{2})(1-(x^2)^{2})} & \frac{1}{1-(x^2)^{2}} & -\frac{(x^2)(x^3)}{(1-(x^2)^{2})(1-(x^3)^{2})}\\
	-\frac{(x^3)(x^1)}{(1-(x^3)^{2})(1-(x^1)^{2})} & -\frac{(x^2) (x^3)}{(1-(x^2)^{2})(1-(x^3)^{2})} & \frac{1}{1-(x^3)^{2}}
	\end{array}
	\right) \\
&= 9k\, H_{x}^{(T)}.
\end{flalign*}
\end{proof}

\section{Discussions \label{sec:Discussions}}

\begin{figure}[t]
\centering{}
\includegraphics[scale=0.7]{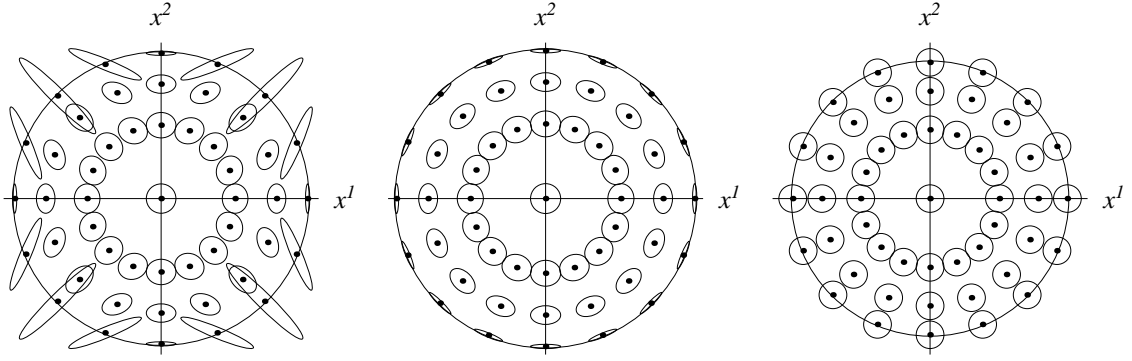}
\caption{Indicatrices for several typical weights $H_x$, 
where $H_x=H_x^{(T)}$ (left), 
$H_x=J_x$ (middle), 
and $H_x=I$ (right).
\label{fig:weight}}
\end{figure}

Let us investigate the properties of the weight $H_x^{(T)}$ that is optimal for the tomography. 
We first regard 
a weight $H_x$ as a metric tensor on the tangent space $\T_{\tau_x} \S$ at $x\in\X$,
and let us plot the indicatrix, the set of end points of tangent vectors $\v \in \T_x \S$ centered at $x$ satisfying $^{t} \v H_x \v = 1$. 
Figure \ref{fig:weight} 
shows the indicatrices on the $x^1 x^2$-plane 
for $H_x=H_x^{(T)}$ (left), 
$H_x=J_x$ (middle),
and $H_x=I$ (right).
Obviously $H_x^{(T)}$ is not rotationally symmetric, 
and is awkwardly distorted 
when $x=(x^1,x^2,x^3)\in\X$ is 
off the coordinate axes. 
This means 
that the tomography depends highly on the choice of the coordinate axes. 
Actually, 
an estimation scheme should be independent of the choice of the coordinate axes
because
their choice is completely arbitrary. 
It is therefore natural to adopt a rotationally symmetric weight $H_x$ which satisfies $U^{*} H_{(U x)} U = H_x$ for $U \in SO(3)$. 

Any rotationally symmetric weight can be represented by
\begin{equation}
H_x^{(f,g)}:=f(r) I + (g(r)-f(r)) \frac{1}{r^2}\ket{x} \bra{x}, \label{eq:roWeight}
\end{equation}
for $x \neq 0$ where $f,g$ are functions on $(0,1)$ such that $f(r)>0$ and $g(r)>0$ 
(see Appendix \ref{sec:RotationallySymmetricWeight}). 
Given a weight $H_{x}=H_x^{(f,g)}$, 
let $M^{(f,g)} \in \M(\H)$ be the corresponding optimal measurement given by \eqref{eq:optMeasurement},
and let $c_x:=\Tr H_x^{(f,g)} g_x(M^{(f,g)})^{-1}$ and $c_x^{(T)}:=\Tr H_x^{(f,g)} g_x(M^{(T)})^{-1}$. 
It then follows from \eqref{eq:qbitBound} and  \eqref{eq:gM} that
\begin{eqnarray}
c_x &=& \left( \Tr \sqrt{\sqrt{J_x^{-1}}H_x^{(f,g)}\,\sqrt{J_x^{-1}}} \right)^2 \nonumber \\
&=& \left( 2 \sqrt{f(r)} + \sqrt{(1-r^2) g(r)}  \right)^2, \label{eq:roc}
\end{eqnarray}
and 
\begin{equation}
c_x^{(T)}
= 3 (2 f(r)+(1- r^2) g(r)) + 3 t r^2 (g(r)-f(r)), \label{eq:rocT}
\end{equation}
where $t:=1-\frac{(x^1)^4+(x^2)^4+(x^3)^4}{r^4}$. 
Note that $0 \leq t \leq  \frac{2}{3}$, 
and that $t=0$ if and only if $x$ is on one of the coordinate axes, 
and $t=\frac{2}{3}$ if and only if $x$ is parallel to one of the vectors $(1,1,1)$, $(-1,1,1)$, $(1, -1,1)$, and $(1,1,-1)$. 
In addition, 
\begin{eqnarray}
c_x^{(T)}-c_x 
&=& 2 \left(\sqrt{(1-r^2) g(r)} - \sqrt{f(r)} \right)^2 + 3 r^2 \left( g(r) - f(r) \right) t \label{eq:diff1} \\
&=& 2 \left(\sqrt{(1-r^2) f(r)} - \sqrt{g(r)} \right)^2 + 3 r^2 \left(f(r) - g(r) \right) \left( \frac{2}{3}-t \right). \label{eq:diff2}
\end{eqnarray}
By considering the cases $g(r) \geq f(r)$ and $f(r) > g(r)$ separately, 
we conclude that $c_x^{(T)} \geq c_x$ for any rotationally symmetric weight $H_x^{(f,g)}$. 

For example, when $H_x^{(f,g)}=J_x$, 
for which $f(r)=1$ and $g(r)=\frac{1}{1-r^2}$, 
we see that $g(r)-f(r) \rightarrow \infty $ as $r \rightarrow 1$, 
so that $c_x^{(T)}$ becomes much larger than $c_x$. 
On the other hand, 
when $H_x^{(f,g)}=I$, 
for which $f(r)=g(r)=1$, 
the second terms in \eqref{eq:diff1} and \eqref{eq:diff2} vanish, 
and the difference $c_x^{(T)}-c_x$ becomes relatively small. 
Figure \ref{Flo:bound} shows the behavior of $c_{r \v}$ (solid) and $c_{r \v}^{(T)}$ (dashed) as functions of radius $r$ 
in the direction $\v=\frac{1}{\sqrt{3}}(1,1,1)^t$ 
for $H_x^{(f,g)}=J_x$ (left) and $H_x^{(f,g)}=I$ (right). 
When $H_x^{(f,g)}=J_x$, 
we see that $c_{r \v}^{(T)}$ diverges as $r \rightarrow 1$, 
while $c_{r \v}$ converges to $9$. 
When $H_x^{(f,g)}=I$, 
on the other hand, 
$c_{r \v}^{(T)}$ and $c_{r \v}$ converge to $6$ and $4$ respectively as $r \rightarrow 1$, 
and their difference is relatively small. 

\begin{figure}[t]
\centering{}
\includegraphics[scale=0.7]{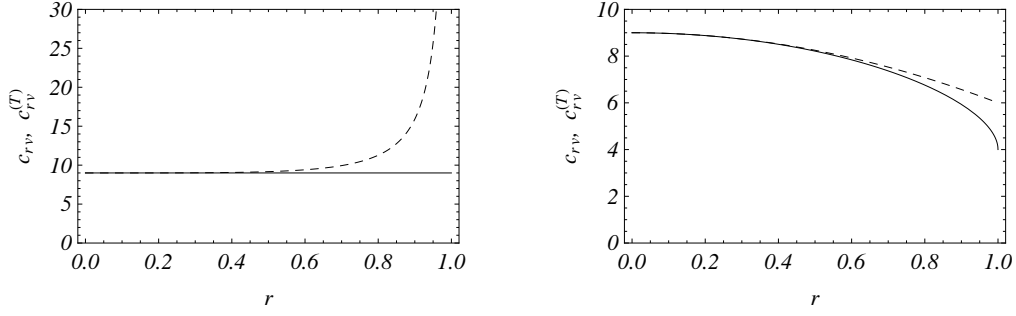}
\caption{
The behavior of $c_{r \v}$ (solid) and $c_{r \v}^{(T)}$ (dashed) as functions of radius $r$ 
in the direction $\v=\frac{1}{\sqrt{3}}(1,1,1)^t$ 
for $H_x^{(f,g)}=J_x$ (left) and $H_x^{(f,g)}=I$ (right). 
}
\label{Flo:bound}
\end{figure}

Now let us make a numerical simulation to compare the asymptotic performance of the tomography and the 
optimal adaptive estimation schemes
for $H_x=J_x$ and $H_x=I$. 
We set the qubit state to be estimated as $\tau_{x_{0}}$ with $x_{0}=(0.55,0.55,0.55)$. 
Since the optimal estimator given in Proposition \ref{prop:qbitEstM} depends on the true value of $x\in\X$, 
we shall invoke an adaptive estimation scheme in evaluating $\Tr H_x g_x(M(x))^{-1}$,
with $M(x)$ being the optimal POVM for $x\in\X$, as follows \cite{nagaokaAdaptive,strong}: 
We begin by choosing $\hat{x}^{(0)}\in\X$ arbitrarily. 
Suppose that $M(\hat{x}^{(0)})$ is applied and that the outcome $n_{1}\in\{1,2,\dots,s\}$ is obtained. 
The maximum likelihood estimator is given by
\[
\hat{x}^{(1)}:=\argmax_{x\in\X}l_{1}(x),
\]
where
\[
l_{1}(x):=\log\Tr\tau(x)\, M_{n_{1}}(\hat{x}^{(0)}).
\]
At the $m$th stage $(m\geq2)$, suppose that $M(\hat{x}^{(m-1)})$ is applied and that the outcome $n_{m}\in\{1,2,\dots,s\}$ is obtained. 
The maximum likelihood estimator at the $m$th stage is given by
$$\hat{x}^{(m)}:=\argmax_{x\in\X}l_{m}(x),$$
where
$$l_{m}(x):=\sum_{i=1}^{m}\log\Tr\tau(x)\, M_{n_{i}}(\hat{x}^{(i-1)}).$$
Because of the strong consistency and the asymptotic efficiency of the adaptive estimation \cite{strong}, 
the sequence $m \times \Tr H_{x_0} V[\hat{x}^{(m)}]$ of the weighted covariances multiplied by $m$ 
converges to $\Tr H_{x_0} g_{x_0}(M(x_0))^{-1}$ as $m \rightarrow \infty$. 
Let us demonstrate this behavior by a numerical simulation. 
We have performed two kinds of numerical simulations in which the weight $H_x$ has been set as $H_x=J_x$ and $H_x=I$. 
These results are shown in the left and the right figure in Figure \ref{Flo:sim},  
where the solid and dashed curves correspond to the adaptive estimation and the tomography, 
and the solid and dashed horizontal lines correspond to the theoretical limits. 
As figures of merit, 
we have plotted in Figure \ref{Flo:sim} the sample averages of
$2m \times B(\tau_{x_0},\tau_{\hat{x}^{(m)}})$, 
where $B(\cdot,\cdot)$ is the Bures distance, 
or
$m \times |x_0-\hat{x}^{(m)}|^2$
instead of 
$m \times \Tr J_{x_0} V[\hat{x}^{(m)}]$ 
or
$m \times \Tr V[\hat{x}^{(m)}]$  
because they are asymptotically equivalent (See Appendix \ref{sec:Bures}). 
The sample averages are calculated by repeating the estimation schemes $1000$ times. 
We see that the sample average of each estimation scheme approaches the corresponding theoretical value,
as $m$ becomes large. 
We further observe that the adaptive estimation scheme is more efficient than the tomography, 
and the difference of their performances
is noticeable when $H_x=J_x$. 
We could conclude that the tomography is not efficient for a rotationally symmetric weight
that is natural in estimating an unknown qubit state.

\begin{figure}[t]
\centering{}
\includegraphics[scale=0.7]{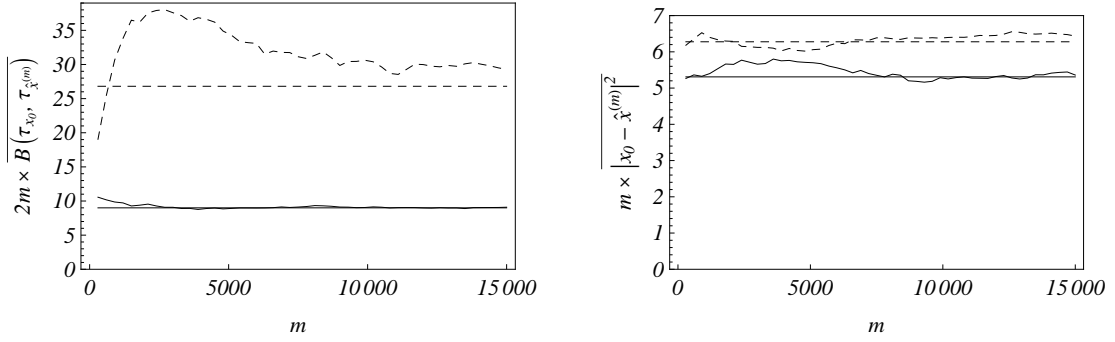}
\caption{
A numerical comparison between the tomography and the optimal adaptive estimation for the weight $H_x$, 
where $H_x$ has been set as $H_x=J_x$ (left) or $H_x=I$ (right). 
The solid and dashed curves correspond to the adaptive estimation and the tomography, respectively, 
and the solid and dashed horizontal lines correspond to the theoretical limit. 
As a figure of merit, 
we have plotted the sample averages of
$2m \times B(\tau_{x_0},\tau_{\hat{x}^{(m)}})$ 
or
$m \times |x_0-\hat{x}^{(m)}|^2$
instead of 
$m \times \Tr H_{x_0} V[\hat{x}^{(m)}]$. 
}
\label{Flo:sim}
\end{figure}

Finally we shall touch upon a generation to  a higher dimensional Hilbert space $\H$. 
Let $q = \dim \H (\geq 3)$ and 
let $\{\ket{e_i^{(\alpha)}}\}_{i=1}^{q}$ 
be an orthonormal basis 
for each $\alpha=1,\dots,q+1$
satisfying $|\braket{e_i^{(\alpha)}}{e_j^{(\beta)}}|^2=\frac{1}{q}$ ($\alpha \neq \beta$) for all $i,j$. 
A finite subset $\{\ket{e_i^{(\alpha)}}\}_{\alpha,i}$ of the Hilbert space $\H$ is called a full set of {\em mutually unbiased bases}. 
It is known that a full set of mutually unbiased bases exists when $q$ is a prime number or the power of a prime \cite{mutual}.
As before, we regard the uniform combination 
$$M^{(T)}:=\frac{1}{q+1}\bigoplus_{\alpha=1}^{q+1} M^{(\alpha)}$$
of the PVMs $M^{(\alpha)}:=( \ket{e_1^{(\alpha)}} \bra{e_1^{(\alpha)}} ,\cdots ,\ket{e_{q}^{(\alpha)}} \bra{e_{q}^{(\alpha)}})\in \M(\H)$ 
as a tomography on $\H$.
Let $\S$ be the set of strictly positive density operators on $\H$,
and let $x=\{x_{\alpha,i}\}$ be an affine parametrization of $\S$ given by
$$
\tau_x=\frac{1}{q}I + \sum_{\alpha=1}^{q+1} \sum_{i=1}^{q-1} x_{\alpha,i} (\ket{e_i^{(\alpha)}} \bra{e_i^{(\alpha)}} - \frac{1}{q}I).
$$
Figure \ref{Flo:high} shows the behavior of 
$c_{r \v}$ (solid) and 
$c_{r \v}^{(T)}$ (dashed) as functions of $r$
in the direction $\v\in \R^{q^2-1}$ where 
\begin{eqnarray}
& c_x & := \min\{\Tr J_x g_x(M)^{-1} \mid M\in\M(\H)\}, \nonumber \\
& c_x^{(T)} & := \Tr J_x g_x(M^{(T)})^{-1} \nonumber
\end{eqnarray}
with
$v_{11}=1$ and $v_{\alpha i}=0$ $(\alpha \neq 1$ or $i \neq 1)$
for $\dim \H=3$ (left) and $\dim \H=4$ (right). 
We see that the behavior for $\dim \H=3$ and $4$ are almost the same as that for $\dim \H=2$ plotted in Figure \ref{Flo:bound}. 
This observation suggests
that the same non-optimality result would hold for $\dim\H \geq 3$.

\begin{figure}[t]
\centering{}
\includegraphics[scale=0.7]{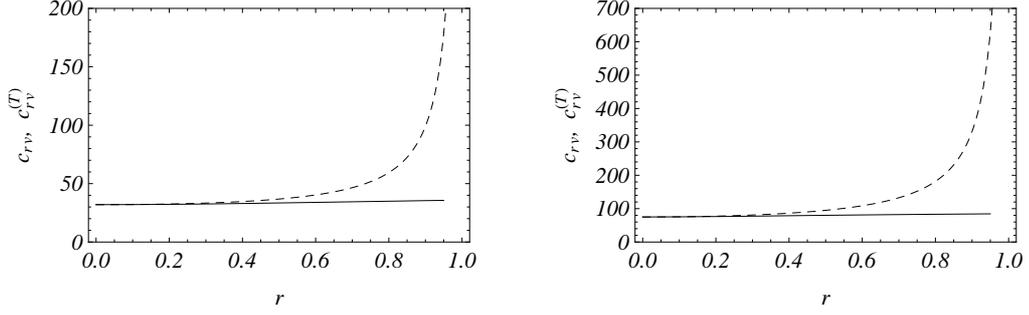}
\caption{
The behavior of 
$c_{r \v}$ (solid) and 
$c_{r \v}^{(T)}$ (dashed) as functions of $r$
in the direction $\v\in \R^{q^2-1}$ where $v_{11}=1$ and $v_{\alpha i}=0$ $(\alpha \neq 1$ or $i \neq 1)$
for $\dim \H=3$ (left) and $\dim \H=4$ (right).
}
\label{Flo:high}
\end{figure}

\section*{Acknowledgment}
The author is grateful to Prof. A. Fujiwara for stimulating discussions and valuable comments.

\appendix 

\section*{Appendices}
\section{Proofs of Propositions \ref{prop:qbitest} and \ref{prop:qbitEstM}  \label{sec:qbitEstProof}}

In this appendix, we give simple proofs of
Propositions \ref{prop:qbitest} and \ref{prop:qbitEstM}
for the reader's convenience. 
We start with some lemmas which hold for an arbitrary finite dimensional Hilbert space $\H$.
Let us define the inner product $\inner{\cdot}{\cdot}_{\theta}$ on $\L(\H)$,  as 
\[
	\inner {A}{B}_{\theta}:=\frac{1}{2} \Tr \rho_{\theta} (A^{*}B+BA^{*}).
\]
Then we can rewrite $g_{\theta}(M)$ by SLD as follows:
\[
	g_{\theta}(M)=
	\left[
	\sum_{x}\frac{\inner{L_{i}}{M_{x}}_{\theta}\inner{L_{j}}{M_{x}}_{\theta}}{\inner I{M_{x}}_{\theta}}
	\right]
	_{1\leq i,j\leq d},
\]
Further we can also rewrite $J_{\theta}$ as 
$J_{\theta}=[\inner{L_i}{L_j}]_{ij}$.
Let us define $\hat{L}^i$ as  \eqref{eq:Lhat}.
Let us define
\[
	\hat{g}_{\theta}(M)=
	\left[
	\sum_{x}\frac{\inner{\hat{L}^{i}}{M_{x}}_{\theta}\inner{\hat{L}^{j}}{M_{x}}_{\theta}}{\inner I{M_{x}}_{\theta}}
	\right]
	_{1\leq i,j\leq d}.
\]

\begin{lem}
$\{\hat{L}^i\}_i \cup \{I\}$ is orthonormal with respect to $\inner{\cdot}{\cdot}_{\theta}$.
\end{lem}
\begin{proof}
\[
\inner{\hat{L}^{i}}{\hat{L}^{j}}_{\theta}=\sum_{s,t}K^{is}K^{jt}\inner{L_{s}}{L_{t}}_{\theta}=\sum_{s,t}K^{is}J_{\theta,st}(K^{*})^{tj}=(U^{-1}\sqrt{J_{\theta}^{-1}}J_{\theta}\sqrt{J_{\theta}^{-1}}U)^{ij}=\delta^{ij}.
\]
Further
\[
\inner{\hat{L}^i}{I}_{\theta}=\sum_{s}K^{is}\inner{L_s}{I}_{\theta}=\sum_{s}K^{is}\Tr \rho_{\theta}L_s
=\sum_{s}K^{is}\Tr \partial_i\rho_{\theta}=0,
\]
then
\[
\inner{I}{I}_\theta=\Tr \rho_\theta=1.
\]
\end{proof}

\begin{lem}\label{lem:Trg}
It holds that
\[
	\Tr \hat{g}_{\theta}(M)\leq\dim\H-1,
\]
for all $M\in\M(\H)$. 
\end{lem}

\begin{proof}
\begin{eqnarray}
	\Tr \hat{g}_{\theta}(M) 
	& = & 
	\sum_{x}\frac{\sum_{i=1}^{d}\inner{\hat{L}^i}{M_{x}}_{\theta}^{2}}{\inner I{M_{x}}_{\theta}}
	=\sum_{x}\left(\sum_{i=1}^{d}\frac{\inner{\hat{L}^i}{M_{x}}_{\theta}^{2}+\inner I{M_{x}}_{\theta}^{2}}{\inner I{M_{x}}_{\theta}}-\inner I{M_{x}}_{\theta}\right)\nonumber \\
	 & \leq & \sum_{x}\left(\frac{\inner{M_{x}}{M_{x}}_{\theta}}{\inner I{M_{x}}_{\theta}}-\inner I{M_{x}}_{\theta}\right)\label{eq:bessel}\\
	 & = & \sum_{x}\frac{\inner{M_{x}}{M_{x}}_{\theta}}{\inner I{M_{x}}_{\theta}}-1\leq\sum_{x}\Tr M_{x}-1\label{eq:showrtz1}\\
	 & = & \Tr I-1=\dim\H-1.\nonumber
\end{eqnarray}
Inequality \eqref{eq:bessel} follows from Bessel's inequality, and 
inequality \eqref{eq:showrtz1} from 
\[
	\inner {I}{M_{x}}_{\theta}\,\Tr M_{x}
	=(\Tr\rho_{\theta}M_{x})\,(\Tr M_{x}) \geq \Tr\rho_{\theta}M_{x}^{2}
	=\inner{M_{x}}{M_{x}}_{\theta}.
\]
\end{proof}

\begin{lem}\label{lem:Convexity}
    Let  $g_{\theta}(\M(\H)):=\{g_{\theta}(M) \mid M \in \M(\H)\}$. Then
	$g_{\theta}(\M(\H))$ is a convex set.
	Similarly,  $\hat{g}_{\theta}(\M(\H))$ is also a convex set.
\end{lem}
\begin{proof}
	Let $M^{(1)},M^{(2)}\in\M(\H)$ and let $0\leq p\leq1$. 
	Then we see
	\begin{multline}
		g_{\theta}(p M^{(1)} \oplus (1-p) M^{(2)})_{ij} \\ 
		=\sum_{x} \frac{p^2 \inner{L_i}{M^{(1)}_{x}}_{\theta} \inner{L_j}{M^{(1)}_{x}}_{\theta}}
		{p \inner{I}{M^{(1)}_{x}}_{\theta}} 
		+ \sum_{y} \frac{(1-p)^2 \inner{L_i}{M^{(2)}_{y}}_{\theta} \inner{L_j}{M^{(2)}_{y}}_{\theta}}
		{(1-p) \inner{I}{M^{(2)}_{y}}_{\theta}}\\
		=\sum_{x} p \frac{\inner{L_i}{M^{(1)}_{x}}_{\theta} \inner{L_j}{M^{(1)}_{x}}_{\theta}}
		{\inner{I}{M^{(1)}_{x}}_{\theta}} 
		+ \sum_{y} (1-p) \frac{\inner{L_i}{M^{(2)}_{y}}_{\theta} \inner{L_j}{M^{(2)}_{y}}_{\theta}}
		{\inner{I}{M^{(2)}_{y}}_{\theta}}\\
		=pg_{\theta}(M^{(1)})_{ij}+(1-p)g_{\theta}(M^{(2)})_{ij}.\label{eq:affine}
	\end{multline}
	This implies that any convex combination of $g_{\theta}(M^{(1)})$ and $g_{\theta}(M^{(2)})$ belongs to $g_{\theta}(\M(\H))$.
\end{proof}

Now we restrict ourselves to the case when $\dim \H=2$. 
In this case it is necessary that $1 \leq  d \leq 3$.

\begin{lem}\label{lem:vv}
Given $\v=(v_{1},\dots,v_{d})^{t}\in\R^{d}$ 
such that $\left|\v\right|=1$, then 
\begin{equation}
\hat{g}_{\theta}(M^{(\v)})=\ket{\v}\bra{\v} \label{eq:vv},
\end{equation}
where $M^{(\v)}$ is a projection-valued measurement given by the spectral decomposition of $L_{\v}:=\sum_{i=1}^{d}v_{i}\hat{L}^i.$
\end{lem}
\begin{proof}
\begin{align}
\bra{\v}\hat{g}_{\theta}(M^{(\v)})\ket{\v} &=  
\sum_{x}\sum_{i,j}v_{i}v_{j}
\frac{\inner{\hat{L}^i}{M_{x}^{(\v)}}_{\theta} \inner{\hat{L}^j}{M_{x}^{(\v)}}_{\theta}}
{\inner I{M_{x}^{(\v)}}_{\theta}}\nonumber \\
 &=  \sum_{x}\frac{\inner{L_{\v}}{M_{x}^{(\v)}}_{\theta}^{2}}{\inner {I}{M_{x}^{(\v)}}_{\theta}}
=  \sum_{x}\frac{\inner{L_{\v}}{M_{x}^{(\v)}}_{\theta}^{2}}{\inner {M_{x}^{(\v)}}{M_{x}^{(\v)}}_{\theta}} 
= \sum_x \inner{L^{\v}}{\tilde{M}_x^{(\v)}}_{\theta}^2 \nonumber\\
&\leq  \inner{L_{v}}{L_{v}}_{\theta}=1 \label{bessel2},
\end{align}
where $\tilde{M}_{x}^{(\v)}:=M_{x}^{(\v)}/ \sqrt{\inner{M_{x}^{(\v)}}{M_{x}^{(\v)}}}$. 
Because $\{\tilde{M}_{x}^{(\v)}\}_x$ is orthonormal with respect to $\inner{\cdot }{\cdot }_{\theta}$,
the inequality \eqref{bessel2} follows from Bessel's inequality.
Further by definition, $L_{\v} \in \spann \{\tilde{M}_{x}^{(\v)}\}_x$. Therefor
\begin{equation}
\bra{\v}\hat{g}_{\theta}(M^{(\v)})\ket{\v} =1. \label{eq:upper}
\end{equation}
According to Lemma \ref{lem:Trg}, 
\begin{equation}
\Tr \hat{g}_{\theta}(M^{(\v)}) \leq  \dim\H-1 =1. \label{eq:trg1}
\end{equation}
We can conclude \eqref{eq:vv} 
from \eqref{eq:upper} and \eqref{eq:trg1} and $\hat{g}_{\theta}(M^{(\v)}) \geq 0$.
\end{proof}

\begin{lem}\label{pro:seteq}
Let $M^{+}(d,\R)$ be the set of $d \times d$ real positive semi definite matrices. Then
\[
\hat{g}_{\theta}(\M(\H))=\left\{ G\in M^{+}(d,\R) \mid \,\Tr G \leq 1\right\} .
\]
\end{lem}
\begin{proof}
According to Lemma \ref{lem:vv}, for any $\boldsymbol{v}=(v_{1},\dots,v_{d})^{t}\in\mathbb{R}^{d}$ such that $\left|\boldsymbol{v}\right|=1$,
\[
\ket{v}\bra{v} \in \hat{g}_{\theta}(\M(\H)).
\]
We further observe that $0 \in \hat{g}_{\theta}(\M(\H))$ because the POVM $M^{(0)}:=(I)$ provides no information. 
Then we see from Lemma \ref{lem:Convexity} that 
\[
\hat{g}_{\theta}(\M(\H)) 
\supset co(\left\{ \ket{v}\bra{v} \biggm| \, v\in\mathbb{R}^{d},\left|\boldsymbol{v}\right|=1\right\} \cup\{0\})
=\left\{ G\in\realmatrix|\,\Tr G\leq1\right\}.
\]
The converse inclusion follows from Lemma \ref{lem:Trg}.
\end{proof}

\begin{lem}\label{lem:fisherset}
\[
g_{\theta}(\M(\H))=\left\{ \sqrt{J_{\theta}}G\sqrt{J_{\theta}} \biggm| G\in\realmatrix,\Tr G\leq1\right\} .
\]
\end{lem}
\begin{proof}
$$
\hat{g}_{\theta}(M)_{ij}=\sum_{st} K^{is} K^{jt} g_{\theta}(M)_{st} = \sum_{st} K^{is} g_{\theta}(M)_{st} (K^*)^{tj},
$$
thus
$$
\hat{g}_{\theta}(M)=U^{-1} \sqrt{J_{\theta}^{-1}} g_{\theta}(M) \sqrt{J_{\theta}^{-1}} U.
$$
Therefore
$$
\sqrt{J_{\theta}} U \hat{g}_{\theta}(M) U^{-1} \sqrt{J_{\theta}} = g_{\theta}(M).
$$
It follows from lemma \ref{pro:seteq} that
$$
g_{\theta}(\M(\H))
=\left\{ \sqrt{J_{\theta}} U G U^{-1} \sqrt{J_{\theta}} \biggm| G\in \hat{g}_{\theta}(\M(\H))   \right\}
=\left\{ \sqrt{J_{\theta}}G\sqrt{J_{\theta}} \biggm| G\in \realmatrix,\Tr G\leq1\right\} .
$$
\end{proof}
\begin{lem}\label{lem:Lagrange}
Given $S\in M^{+}(d,\R)$ such that $S>0$,
\[
\min\left\{ \Tr SG^{-1};G\in\realmatrix,\,\Tr G=1\right\} =(\Tr\sqrt{S})^{2}.
\]
Only if $G=\sqrt{S}/(\Tr\sqrt{S})$ then $\Tr SG^{-1}=(\Tr\sqrt{S})^{2}$.
\end{lem}
\begin{proof}
For $G=\left(g_{ij}\right)_{1\leq i,j\leq d}$ , let $f(G):=\Tr(SG^{-1})+\lambda(\Tr G-1)$
where $\lambda$ is a Lagrange multiplier. Then\[
\frac{\partial f}{\partial G_{ij}}=\Tr\left[S(-G^{-1}\frac{\partial G}{\partial G_{ij}}G^{-1})\right]+\lambda\delta_{ij}=-\bra{e_{j}}G^{-1}SG^{-1}\ket{e_{i}}+\lambda\delta_{ij}=0\]
where $\left\{ e_{i}\right\} _{1\leq i\leq d}$ is the standard CONS
of $\R^{d}$. Thus
\begin{eqnarray*}
G^{-1}SG^{-1} & = & \lambda I
\end{eqnarray*}
from which\begin{eqnarray*}
G & = & \frac{\sqrt{S}}{\sqrt{\lambda}}\end{eqnarray*}
and\[
\lambda=\left(\Tr\sqrt{S}\right)^{2}\]
because of $\Tr G=1$. As a consequence\[
\min_{G}\Tr(SG^{-1})=\Tr(\lambda G)=\lambda=(\Tr\sqrt{S})^{2}.\]
\end{proof}

\begin{proof}[\bf{Proof of Proposition \ref{prop:qbitest}}]
According to Lemma \ref{lem:fisherset} and Lemma \ref{lem:Lagrange},

\begin{eqnarray*}
\min_{M\in\M(\H)}\Tr H_{\theta}\, g_{\theta}(M)^{-1} & = & \min\{\Tr H_{\theta}\sqrt{J_{\theta}^{-1}}G^{-1}\sqrt{J_{\theta}^{-1}}\,|\, G\in M^{+}(d,\mathbb{R}),\Tr G=1\}\\
 & = & \min\{\Tr\sqrt{J_{\theta}^{-1}}H_{\theta}\sqrt{J_{\theta}^{-1}}G^{-1}\,|\, G\in M^{+}(d,\mathbb{R}),\Tr G=1\}\\
 & = & \left(\Tr R_{\theta}\right)^{2}.\end{eqnarray*}
When $\Tr H_{\theta}\, g_{\theta}(M)^{-1}$ achieves the minimum,
$$G = \frac{R_{\theta}}{\Tr R_{\theta}}.$$
thus 
$$g_{\theta}(M) = \sqrt{J_{\theta}}G\sqrt{J_{\theta}} = \frac{\sqrt{J_{\theta}}R_{\theta}\sqrt{J_{\theta}}}{\Tr R_{\theta}}.$$
\end{proof}

\begin{proof}[\bf{Proof of Proposition \ref{prop:qbitEstM}}]
Assume that $d=3$. According to \eqref{eq:affine} and Lemma \ref{lem:vv},
\begin{eqnarray*}
g_{\theta}(M) & = & \sqrt{J_\theta} U \hat{g}_{\theta}(M) U^{-1} \sqrt{J_\theta}
=\sqrt{J_\theta} U \{p_{1} \hat{g}_{\theta}(M^{(1)})+p_{2} \hat{g}_{\theta}(M^{(2)})+p_{3} \hat{g}_{\theta}(M^{(3)})\} U^{-1} \sqrt{J_\theta}\\
 & = & \sqrt{J_\theta} U \{p_{1}\left(\begin{array}{ccc}
1 & 0 & 0\\
0 & 0 & 0\\
0 & 0 & 0\end{array}\right)+p_{2}\left(\begin{array}{ccc}
0 & 0 & 0\\
0 & 1 & 0\\
0 & 0 & 0\end{array}\right)+p_{3}\left(\begin{array}{ccc}
0 & 0 & 0\\
0 & 0 & 0\\
0 & 0 & 1\end{array}\right)\} U^{-1} \sqrt{J_\theta}\\
 & = & \sqrt{J_\theta} U \frac{S}{\Tr S} U^{-1} \sqrt{J_\theta}
=\sqrt{J_{\theta}}\frac{USU^{-1}}{\Tr S}\sqrt{J_{\theta}}=\frac{\sqrt{J_{\theta}}R_{\theta}\sqrt{J_{\theta}}}{\Tr R_{\theta}}.
\end{eqnarray*}
When $d=1$ or $2$, we can prove this in a similar way.
\end{proof}

\section{Rotationally symmetric weight \label{sec:RotationallySymmetricWeight}}
In this appendix, we show that 
any rotationally symmetric weight  
is represented in the form
\begin{equation}
H_x:=f(r) I + (g(r)-f(r)) \frac{1}{r^2} \ket{x} \bra{x},\label{eq:roWeight2}
\end{equation}
for $x \neq 0$ where $f,g$ are functions on $(0,1)$ such that $f(r)>0$ and $g(r)>0$. 

Given $x \in \X$ ($x \neq 0$) arbitrarily, 
let $e_1,e_2,e_3$ be an orthonormal basis of $\R^3$ with $e_3=\frac{\ket{x}}{|x|}$, 
and let $V \in SO(3)$ be any rotation about $e_3$-axis. 
Since
\begin{equation}
V^{*} H_{x} V = V^{*} H_{(V x)} V = H_x, \label{eq:roEq1}
\end{equation}
$H_x$ and $V$ are simultaneously diagonalized, and $e_3$ is one of their common eigenvectors. 
Other eigenvalues of $H_x$ must be degenerate because $V$ is any rotation about $e_3$-axis. 
Then $H_x$ should be represented as
\begin{eqnarray}
H_x &=& \hat{f}(x) \ket{e_1}\bra{e_1} + \hat{f}(x) \ket{e_2}\bra{e_2} + \hat{g}(x) \ket{e_3}\bra{e_3} \nonumber \\
       &=& \hat{f}(x) I + (\hat{g}(x)-\hat{f}(x)) \frac{1}{r^2}\ket{x} \bra{x} \label{eq:roHikaku1}.
\end{eqnarray}
Let $U \in SO(3)$ be any rotation. It follows that
\begin{eqnarray}
U^{*} H_{(U x)} U &=& U^{*} \left[ \hat{f}(U x) I + (\hat{g}(U x)-\hat{f}(U x)) \frac{1}{r^2}\ket{U x} \bra{U x} \right] U \nonumber \\
&=& \hat{f}(U x) I + (\hat{g}(U x)-\hat{f}(U x)) \frac{1}{r^2}  \ket{U^* U x} \bra{U^* U x} \nonumber \\
&=& \hat{f}(U x) I + (\hat{g}(U x)-\hat{f}(U x)) \frac{1}{r^2}  \ket{x} \bra{x} \label{eq:roHikaku2}.
\end{eqnarray}
We see that it follows $\hat{f}(x)=\hat{f}(U x)$ and $\hat{g}(x)=\hat{g}(U x)$ for any $U \in SO(3)$
by comparing \eqref{eq:roHikaku1} and \eqref{eq:roHikaku2}. 
Therefore $\hat{f}$ and $\hat{g}$ must be represented by $\hat{f}(x)=f(|x|)$ and $\hat{g}(x)=g(|x|)$.

\section{Bures distance and quantum Fisher information matrix \label{sec:Bures}}
The Bures distance between two states $\rho$ and $\sigma$ is 
defined by
$$B(\rho,\sigma):=4\left( 1-\Tr \sqrt{\sqrt{\rho}\sigma\sqrt{\rho}} \right).$$ 
It is known that 
\begin{equation}
B(\tau_x,\tau_{x+dx})=\frac{1}{2} \sum_{ij} J_{x,ij} dx^i dx^j + O(|dx|^3) \label{eq:nearBures} 
\end{equation}
when $|dx|$ is sufficiently small. 
Given an estimator $(M,\hat{x})$ that is locally unbiased at $x_0 \in \X$, 
it follows from \eqref{eq:nearBures} that
\begin{eqnarray*}
\Tr J_{x_0} V_{x_0}[M,\hat{x}] &=& E_{x_0}[M,\sum_{ij} J_{x_0,ij} (\hat{x}^i - x_0^i) (\hat{x}^j - x_0^j)]\\
&=& E_{x_0}[M,2 B(\tau_{x_0},\tau_{\hat{x}}) + O(|\hat{x}-x_0|^3)]. 
\end{eqnarray*}

\end{document}